\documentclass[conference]{IEEEtran}
\IEEEoverridecommandlockouts
\usepackage{cite}
\usepackage{amsmath,amssymb,amsfonts}
\usepackage[normalem]{ulem}
\usepackage{amsthm}
\usepackage{graphicx}
\usepackage{textcomp}
\usepackage{xcolor}
\usepackage{caption}
\usepackage{booktabs}
\usepackage{algorithm}
\usepackage[noend]{algpseudocode}
\usepackage{url}

\newif\iflong
\longtrue

\def\BibTeX{{\rm B\kern-.05em{\sc i\kern-.025em b}\kern-.08em
    T\kern-.1667em\lower.7ex\hbox{E}\kern-.125emX}}
\begin{document}

\newtheorem{theorem}{Theorem}[section]
\newtheorem{conjecture}[theorem]{Conjecture}
\newtheorem{proposition}[theorem]{Proposition}
\newtheorem{lemma}[theorem]{Lemma}
\newtheorem{corollary}[theorem]{Corollary}
\newtheorem{example}[theorem]{Example}
\newtheorem{definition}[theorem]{Definition}

\title{Asymmetric Differential Privacy\thanks{Proofs are omitted due to space limitations, and we refer to the full version for them https://arxiv.org/abs/2103.00996.}}

\author{\IEEEauthorblockN{Shun Takagi}
\IEEEauthorblockA{\textit{Kyoto University} \\
Japan\\
takagi.shun.45a@st.kyoto-u.ac.jp}
\and
\IEEEauthorblockN{Fumiyuki Kato}
\IEEEauthorblockA{\textit{Kyoto University} \\
Japan \\
kato.fumiyuki.68z@st.kyoto-u.ac.jp}
\and
\IEEEauthorblockN{ Yang Cao}
\IEEEauthorblockA{\textit{Kyoto University}\\
Japan \\
yang@i.kyoto-u.ac.jp}
\and
\IEEEauthorblockN{Masatoshi Yoshikawa}
\IEEEauthorblockA{\textit{Kyoto University} \\
Japan \\
yoshikawa@i.kyoto-u.ac.jp}
}

\maketitle

\begin{abstract}
Differential privacy (DP) is getting attention as a privacy definition when publishing statistics of a dataset.
This paper focuses on the limitation that DP inevitably causes two-sided error, which is not desirable for epidemic analysis such as “how many COVID-19 infected individuals visited location A”.
For example, consider publishing misinformation that many infected people did not visit location A, which may lead to miss decision-making that expands the epidemic.
To fix this issue, we propose a relaxation of DP, called asymmetric differential privacy (ADP).
We show that ADP can provide reasonable privacy protection while achieving one-sided error.
Finally, we conduct experiments to evaluate the utility of proposed mechanisms for epidemic analysis using a real-world dataset, which shows the practicality of our mechanisms.
\end{abstract}


\section{Introduction}
\label{sec-intro}
Differential privacy~(DP)~\cite{dwork2006differential} is becoming a gold standard privacy notion.
The US Census adopted DP when publishing the 2020 Census results~\cite{bureaumap}, and IT companies such as Google~\cite{erlingsson2014rappor}, Apple~\cite{apple}, Microsoft~\cite{ding2017collecting}, and Uber~\cite{johnson2018towards} are using DP to protect privacy while collecting data.
This widespread adoption comes from mathematical rigorousness under the assumption that an adversary has any knowledge about an individual's data in a dataset.

However, the rigorousness of DP causes utility limitations.
We focus on a utility limitation: \textit{two-sided error}, which occurs when solving a decision problem (i.e., a query whose answer is binary $\{{\rm True}, {\rm False}\}$) by a randomized algorithm such as Monte Carlo algorithm~\cite{kudelic2016monte}.
To the best of our knowledge, this paper for the first time considers this type of error in DP\@.
\textit{Two-sided} stands for the characteristic that both potential answers True and False may be wrong.
In other words, an algorithm causes false positives and false negatives.
On the other hand, \textit{one-sided} stands for the characteristic that either false positive or false negative occurs.
Without loss of generality, hereafter, we say that error is one-sided if only a false negative occurs.
That is, if the error is one-sided, there is a guarantee of accuracy for the answer of True.

For an example of an issue of two-sided error, we consider publishing statistics of trajectories of infected people as measures to control epidemic diseases such as COVID-19~\cite{park2020information} as shown in Figure~\ref{fig:motivative_example}.
Concretely, we query whether a location is safe or not with respect to epidemic disease by a threshold query (i.e., a decision problem that queries whether a count is under the threshold or not).
Although one may use DP to answer the query with privacy protection, answering a decision problem by any DP mechanism always causes \textit{two-sided} error.
This means that even if published information says safe (dangerous) for a target location, the location may not be safe (not be dangerous) due to the noise.
This published information may lead the expansion of the epidemic.
Therefore, published information with two-sided error may not be appropriate for high-risk scenarios such as epidemic analysis.

\begin{figure}[t]
  \centering\includegraphics[width=\hsize]{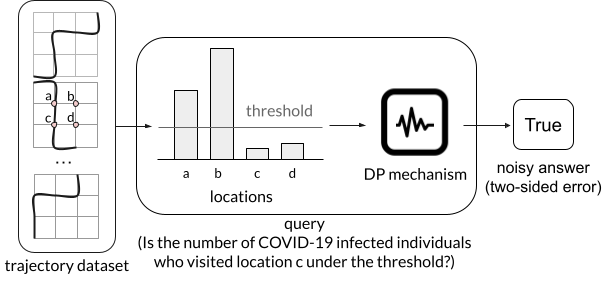}
  \caption{\label{fig:motivative_example}The example of answering a decision problem by a DP mechanism.
  We consider using a trajectory dataset of infected people to publish statistics to support decision-making during a pandemic.
  First, we derive the histogram by counting trajectories that include each target location (here, a, b, c, d).
  Then, we mark the location as “safe” if the count of infected individuals who visited that location is less than a given threshold; otherwise, we mark it as “dangerous”.
  However, the differentially-private answers to such a query have two-sided error.}
\end{figure}

More concretely, in Figure~\ref{fig:motivative_example}, we query "is location $c$ under the threshold?".
The true answer is True, but any $\varepsilon$-DP mechanism outputs False with some probability (the existence of false negative).
Similarly, for the query "is location $b$ under the threshold?", the true answer is False, but any $\varepsilon$-DP mechanism outputs True with some probability (the existence of a false positive).
Therefore, the error is two-sided.
In this example, the false positive involves a high risk with respect to the epidemic disease, but the false negative does not.
Therefore, we attempt to achieve one-sided error with formal privacy protection.

\begin{figure*}[t]
 \begin{minipage}{0.24\hsize}
 \caption*{Ground truth}
  \centering\includegraphics[width=0.9\hsize]{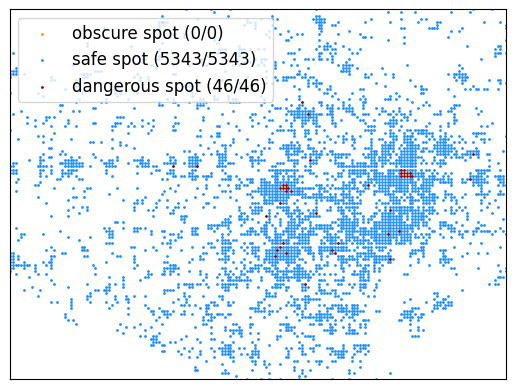}
 \end{minipage}
 \begin{minipage}{0.24\hsize}
 \caption*{$\varepsilon$-DP}
  \centering\includegraphics[width=0.9\hsize]{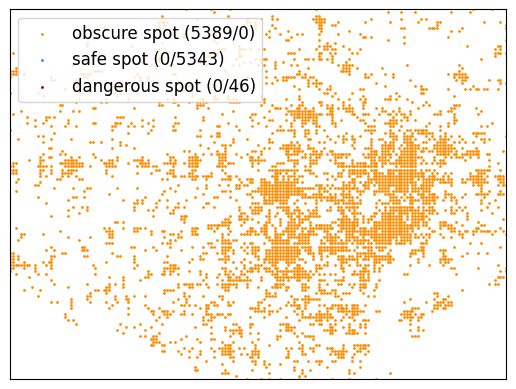}
 \end{minipage}
  \begin{minipage}{0.24\hsize}
  \caption*{$(\varepsilon,\delta)$-DP}
  \centering\includegraphics[width=0.9\hsize]{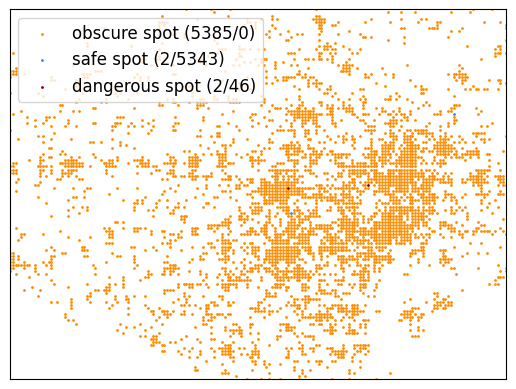}
 \end{minipage}
  \begin{minipage}{0.24\hsize}
  \caption*{ADP}
  \centering\includegraphics[width=0.9\hsize]{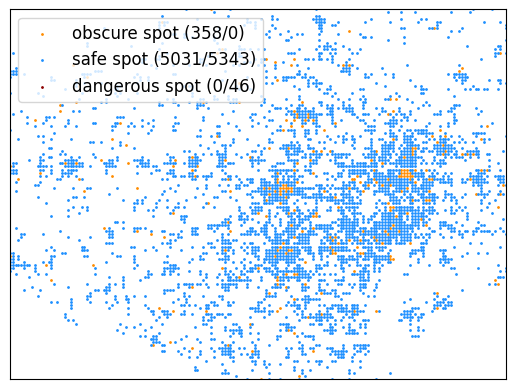}
 \end{minipage}
 \caption{Example of answers of states for $5{,}835$ target locations at ${\rm Dec}/22{\rm th}/2013$ $6$ p.m. from trajectories of $10,338$ users via no privacy (ground truth), $\varepsilon$-DP, $(\varepsilon,\delta)$-DP, and $(\varepsilon,p)$-ADP. 
 The labels $(a/b)$ represent (output number/true number) (e.g., $(\varepsilon,\delta)$-DP outputs only $4$ right answers out of $5{,}343$ safe locations).
 We set $\varepsilon=1$, $\delta=10^{-4}$, and threshold$=3$.\label{fig:result_example}}
\end{figure*}

To this end, we propose Asymmetric Differential Privacy (ADP), which is derived from a necessary condition to achieve one-sided error.
The key to the definition of ADP is a new neighboring relationship, which may seem trivial; however, there is no existing privacy definition that expresses the reasonable privacy protection for the problem of Figure~\ref{fig:motivative_example}, so we need to introduce ADP\@.
Moreover, it is worth noting that our neighboring relationship has interpretability.

Based on the interpretability, we show whether our mechanisms that follow ADP provide reasonable privacy protection for the query "is the target location safe? (i.e., under the threshold)".
We focus on that an answer of Yes or No tends to have biased sensitivity.
For example, No:="one did \textbf{not} visit the target location" is relatively non-sensitive comparing with Yes:="one visited the target location".
From this conception, we show that our mechanisms provide reasonable privacy protection under the assumption that the information "a trajectory does \textbf{not} include the target location" is non-sensitive.
Conversely, we show in a similar vein that to answer the query "is the target location dangerous? (i.e., above the threshold)" with one-sided error, ADP requires the hard assumption that the information "a trajectory includes the target location" is non-sensitive.
Therefore, this paper mainly considers the former query.

Following ADP, we propose mechanisms.
First, we propose a primitive mechanism which is the ADP version of the Laplace mechanism~\cite{dwork2006calibrating}, called the asymmetric Laplace mechanism (ALap).
We show that ALap achieves one-sided error.
However, adapting ALap to our problem causes an issue when the sensitivity of the query is high.
Like the Laplace mechanism~\cite{dwork2006calibrating}, the accuracy linearly decreases as the number of target locations increases.
To solve this problem, we propose the sanitized Laplace mechanism (SALap), the ADP version of the sparse vector technique~\cite{lyu2016understanding}, which satisfies ADP and is robust to the number of target locations.

Finally, we conducted the experiments with the real-world dataset.
The experiment shows that our mechanism is sufficiently practical for location monitoring with one-sided error.
We show the example of a result in Figure~\ref{fig:result_example}.
These represent the safe, dangerous, and obscure states at locations on the Tokyo map at $6$ p.m. published by mechanisms following no privacy (ground truth), $\varepsilon$-DP, $(\varepsilon,\delta)$-DP, and $(\varepsilon,p)$-ADP under the assumption that "a trajectory does \textbf{not} include the target location" is non-sensitive, which is described by $p$.
A safe spot, dangerous spot, and obscure spot respectively represent a location visited at $6$ p.m. by less than $3$, equal to or larger than $3$, and the unclear number of infected people due to the noise.
While $\varepsilon$-DP and $(\varepsilon,\delta)$ marks obscure for almost all spots, ADP can publish many accurate safe spots.
As described above, our mechanism also does not output accurate mark for dangerous spots due to the privacy restriction.

Our contributions are summarized as follows.
\begin{itemize}
    \item We show the issues to achieve one-sided error by $\varepsilon$-DP, $(\varepsilon,\delta)$-DP, Blowfish privacy~\cite{he2014blowfish}, and one-sided differential privacy~\cite{doudalis2020one}.
    \item We propose ADP to solve the issues.
    We show that ADP provides reasonable privacy protection for publishing safe information with one-sided error under the assumption that information "one did \textbf{not} visit a target location" is non-sensitive.
    \item Following ADP, we propose mechanisms, called ALap and SALap, which achieve one-sided error and are robust to the number of target locations.
    \item Our experiments using the real-world dataset show the practicality of our mechanism for location monitoring.
\end{itemize}

\subsection{Related Work}
The most popular relaxation of DP is $(\varepsilon,\delta)$-DP~\cite{dwork2009differential}, and $\delta$ represents the probability that a mechanism breaks $\varepsilon$-DP.
However, we show that large $\delta$ (e.g., $\delta>10^{-1}$) is required to achieve one-sided error with reasonable setting (see Section~\ref{sec:sufficiency} for detail).
Also, some experts warn that $(\varepsilon,\delta)$-DP allows an unacceptable mechanism\footnote{See https://github.com/frankmcsherry/blog/blob/master/posts/2017-02-08.md for details}.
Therefore, $(\varepsilon,\delta)$-DP is not suitable for our problem.

Another way to relax DP is the same as ours: redefining the definition of a neighboring relationship.
Blowfish privacy~\cite{he2014blowfish} and one-sided differential privacy (OSDP)~\cite{doudalis2020one} belong to this family.
However, we also show that their privacy definitions cause unreasonable privacy leaks to achieve one-sided error for the problem of Figure~\ref{fig:motivative_example}.
We describe the detail in Section~\ref{sec:sufficiency}, but here we briefly state the insufficiencies.
\paragraph*{OSDP}
OSDP attempts to utilize non-sensitive records to improve utility~\cite{doudalis2020one}.
Their motivation is completely different from ours, but OSDP implicitly catches \textit{asymmetric} neighboring relationship (i.e., $X$ is neighboring to $X^\prime$, but $X^\prime$ is not neighboring to $X$).
Therefore, OSDP may solve the limitation of two-sided error in specific cases, but the privacy of OSDP is not expressive enough to tightly relax the privacy protection of DP for one-sided error.
Concretely, OSDP only decides whether certain records are protected or not, and OSDP cannot specify which property of a record is protected.
Therefore, OSDP is not sufficiently flexible to stipulate reasonable privacy protection for one-sided error.

\paragraph*{Blowfish privacy}
Blowfish privacy (and Puffefish privacy) also redefines the neighboring relationship~\cite{he2014blowfish,kifer2014pufferfish}.
Blowfish privacy can specify which property of a record is protected, but the definition cannot express the asymmetric neighboring relationship.
Therefore, Blowfish privacy is also not sufficiently flexible, which leads to unexpected privacy leaks.

The classical mechanism called Mangat's randomized response (Mangat's RR)~\cite{mangat1994improved}, which we can interpret as asymmetrization of the traditional randomized response (i.e., Warner's randomized response~\cite{warner1965randomized}), is useful to understand the privacy of asymmetricity of our neighboring relationship in the local setting for binary value (i.e., a dataset is a record).
Mangat focused on the fact that it is often the case that either answer "Yes" or "No" is non-sensitive.
For example, consider the question "do you have a criminal record?".
The answer "Yes" is seriously sensitive, but the answer "No" is relatively non-sensitive.
Mangat's RR returns "Yes" with probability $1$ when the input is "Yes".
When the input is "No", Mangat's RR returns "Yes" with some probability $1-p$ and "No" with some probability $p$.
When the original answer is "Yes", an adversary cannot know the original answer in the way same as DP, but when the original answer is "No", an adversary may correctly know the original answer.
Our paper also utilizes this characteristic for the question "did you visit a location?".
To express this privacy protection in the more general case (i.e., central and multiple attributes), we propose ADP.

It is noted that several relaxation models of local differential privacy (LDP) implicitly expresses asymmetric neighboring relationship~\cite{murakami2019utility,acharya2019context,gu2020providing,cao2020pglp}.
Utility optimized local DP~(ULDP)~\cite{murakami2019utility} is another privacy definition that expresses an asymmetric neighboring relationship.
However, ULDP is defined in the local setting, so we cannot apply ULDP to our task (i.e., central setting).
Context-aware local DP~\cite{acharya2019context} and input-discriminative local DP~\cite{gu2020providing}, which are generalizations of ULDP\@, both define a different indistinguishability level for each combination of data.
Therefore, their relaxations include our relaxation in the local setting, but for the same reason as for ULDP, we cannot apply them to our task.

\section{Background}
Here, we first introduce DP and related definitions, which are the bases of our proposed notions.
Second, we explain counting query, decision problem, and one-sided error this paper handles.
Finally, we show the problem setting of this paper.

The notations used in this paper are listed in Table~\ref{tab:notation}.

\begin{table}[]
\centering
\begin{tabular}{@{}cl@{}}
\toprule
Symbol                                                       & Meaning                                                                                                        \\ \midrule
$X=(x^1,\dots,x^n )\in\mathcal{X}^n$                                            & A dataset.                                                                           \\
$X\sim_{x,x^\prime} X^\prime$ & $X^\prime$ is constructed by replacing \\ 
 & one record that is $x$ with $x^\prime$. \\
$n$                                                         & The number of records in a dataset                                                                               \\
$[n]$ & A set $\{1,\dots,n\}$.\\
$x\in\mathcal{X}$                                     & A record.                                                                          \\
$z\in\mathcal{Z}$                       & An output of a mechanism.\\ 
$\lambda$                                                    & A random variable.                                                                                             \\
$p^t_\leq:\mathbb{N}\to \{\rm False, True\}$ & A threshold proposition with threshold $t$.\\
$M:\mathcal{X}^n\to\mathcal{Z}$ & A randomized mechanism. \\
$\mathbb{R}$                                     & The universe of a real number.\\ 
$\mathbb{N}$  &The universe of a natural number.                                                            \\
$\varepsilon\in\mathbb{R}^+, \delta\in[0,1]$                                     & Privacy parameters.                                                                                                \\
$f_C:\mathcal{X}^n\to\mathcal{Z}$                                & A counting query with conditions $C$.               \\
$\pi$ & A prior distribution on $\mathcal{X}$.\\
\bottomrule
\end{tabular}
\vspace{10pt}
\caption{Notations used in the paper.}
\label{tab:notation}
\end{table}

\subsection{Differential Privacy}
DP~\cite{dwork2006differential} is a mathematical privacy definition that quantitatively evaluates the privacy protection of a randomized mechanism.
A randomized mechanism $M$ is a randomized function that takes a dataset as input and randomly returns $z\in\mathcal{Z}$.
For preliminary, we introduce the definition of the neighboring relationship of DP.
\begin{definition}[Neighboring relationship]
$X$ is neighboring to $X^\prime$ with respect to $x$ and $x^\prime$ if $X^\prime$ is constructed by replacing one record that is $x$ of $X$ with $x^\prime$.
This is denoted by $X\sim_{x,x^\prime} X^\prime$ or simply $X\sim X^\prime$.
\end{definition}
Then, we introduce the definition of approximate DP denoted by $(\varepsilon,\delta)$-DP, which is the most popular generalization of DP.
\begin{definition}[($\varepsilon, \delta$)-DP]
\label{def:dp}
A randomized mechanism $M$ satisfies ($\varepsilon,\delta$)-DP iff $\forall X, X' \in \mathcal{X}^n$ such that $X\sim X^\prime$ and $\forall S \subseteq \mathcal{Z}$,
\begin{equation}
\label{eq:dp}
  \Pr[M(X)\in S] \leq \mathrm{e}^\varepsilon \Pr[M(X^\prime)\in S] + \delta ,
\end{equation}
where $\mathcal{X}^n$ and $\mathcal{Z}$ are the universe of a dataset and an output, respectively.
Simply $\varepsilon$-DP denotes ($\varepsilon,0$)-DP.
\end{definition}

\subsubsection{Generalization of $\varepsilon$-DP}
\label{subsec:generalization of dp}
We introduce two privacy definitions that generalize $\varepsilon$-DP by changing the definition of the neighboring relationship: Blowfish privacy~\cite{Haney2015Design} and OSDP~\cite{doudalis2020one}.
\paragraph*{Blowfish privacy}
\label{subsubsec:def_blowfish}
This paper introduces the simplified version of Blowfish privacy by Haney et al.~\cite{Haney2015Design}.
\begin{definition}[Blowfish privacy]
Given policy graph $G=(\mathcal{X},E)$, $M$ satisfies $(\varepsilon,G)$-Blowfish privacy if $\forall z\in\mathcal{Z}$, $X,X^\prime$ such that $X\sim_{x,x^\prime} X^\prime$ and $(x,x^\prime)\in E,$
$
\Pr[M(X)=z]\leq\mathrm{e}^\varepsilon \Pr[M(X^\prime)=z].
$
\end{definition}
Blowfish privacy changes the definition of the neighboring relationship of $\varepsilon$-DP using a graph.
Note that we do not consider the dummy value~\cite{Haney2015Design} to hide the existence of a record because we assume that the size of a dataset is public information.

\paragraph*{One-sided differential privacy}
OSDP was proposed by Doudalis et al.~\cite{doudalis2020one} to utilize non-sensitive records.
First, we define a set of non-sensitive records $N$ and policy function $P_N:\mathcal{X}\to \{\rm False, True\}$ such that
$$
P_N(x)=\begin{cases}
{\rm True}\ {\rm if}\ x\in N  \\
{\rm False} \ {\rm otherwise}.
\end{cases}
$$
Then, we have the following definition:
\begin{definition}[OSDP]
$M$ satisfies $(\varepsilon, P_N)$-OSDP iff $\forall z\in\mathcal{Z}, X, X^\prime\in\mathcal{X}^n$ such that $X\sim_{x,x^\prime} X^\prime$ and $P_N(x)={\rm False}$,
$
\Pr[M(X)=z]\leq\mathrm{e}^\varepsilon \Pr[M(X^\prime)=z].
$
\end{definition}
OSDP also changes the definition of neighboring relationship of $\varepsilon$-DP by using policy function $P_N$.

\subsection{Problem Setting: Decision Problem}
\label{subsec:decision problem}
In this paper, as a \textit{decision problem}\footnote{We assume a data-dependent decision problem $q$ (i.e., $\exists X,X^\prime\in\mathcal{X}^n, q(X)\neq q(X^\prime)$).}, we introduce a \textit{threshold proposition} to the answer of a \textit{counting query}.
Here, we describe the definitions of counting query, threshold proposition, and one-sided error.
Then, we describe the problem setting of this paper using these notions.

\subsubsection{Counting Query}
\label{subsec:counting query}
Counting query is one of the most basic statistical queries.
Counting query appears in fractional form, with weights (linear query), or in more complex form, but this paper considers the most simple type of counting query $f_C:\mathcal{X}^n\to\{0,1,\dots,n\}^d$, which counts the numbers of records satisfying each condition $c_1,\dots,c_d$, where $C=(c_1,\dots,c_d)$ is a tuple of propositional functions to a record.

\subsubsection{Threshold proposition}
The threshold proposition is a simple proposition to judge whether the number is above (under) the threshold or not.
\begin{definition}
Threshold proposition $p^t_\leq:\mathbb{N}\to\{\rm False,True\}$ is the propositional function such that
\begin{align}
p^t_\leq(m)= \left\{ \begin{array}{ll}
{\rm True} & {\rm if}\ m \leq t \\
{\rm False} & {\rm otherwise} \\
\end{array} \right.
\end{align}
In the same vein, we also define $p^t_\geq (m)$ to judge whether $m$ is above $t$ or not.
\end{definition}

\subsubsection{$(q,\alpha,\beta)$-one-sided error}
\label{sec:one-sided accuracy}
One-sided error is first defined on the Monte Carlo algorithm~\cite{santha1995monte} that is randomized to answer problems that are difficult with respect to complexity.
We say that an algorithm is one-sided if the output True is always correct (i.e., true-biased).
Then, one-sided error is defined on the one-sided algorithm as the probability that the algorithm outputs wrongly outputs False.

To more deeply explore the characteristic of DP with respect to one-sided error, we generalize one-sided error.
Intuitively, we allow error occurring with small probability $\alpha$ for the answer True.
In this state, if the probability that a mechanism wrongly outputs False is lower than $1-\beta$, we say the mechanism is $(q,\alpha,\beta)$-one-sided.
Formally,
\begin{definition}[$(q,\alpha,\beta)$-one-sided]
Mechanism $M$ is $(q,\alpha,\beta)$-one-sided for $X$ such that $q(X)={\rm True}$ if $\forall X^\prime\in\mathcal{X}^n$ such that $q(X^\prime)={\rm False}$, there exists $p^{\rm san}:\mathcal{Z}\to\{\rm False, True\}$ such that
$
\Pr[p^{\rm san}\circ M(X^\prime)={\rm True}]<\alpha\ {\rm and}
$
$
\Pr[p^{\rm san}\circ M(X)={\rm True}] \geq \beta.
$
\end{definition}
Intuitively, $p^{\rm san}$ works like a \textit{sanitizer}~\cite{blum2013learning} to answer $q$.
We call such $p^{\rm san}$ $(q,\alpha,\beta)$-sanitizer.
We can get to know $q(X)={\rm True}$ with high probability (i.e., $1-\alpha$) when we see output True from a $(q,\alpha,\beta)$-sanitizer.
For example, if $M$ is $(q,0,\beta)$-one-sided for $X$, there is a post-processing function $p^{\rm san}$ such that the answer True of $p^{\rm san}\circ M(X)$ is always true and the one-sided error is lower than $1-\beta$.

\subsubsection{Problem Setting}
\label{subsec:problem_setting}
Formally, the purpose is to consider a mechanism that is one-sided to answer the following queries in a privacy preserving way.
Given target locations $(l_1,l_2,\dots,l_d)$ and time $T$ for $i\in [d],$ about an epidemic disease,
\begin{itemize}
    \item $q^i_\leq = p^t_\leq \circ f_{C^T}(\cdot)_i$: location $l_i$ is not an epicenter?
    \item $q^i_\geq = p^t_\geq \circ f_{C^T}(\cdot)_i$: location $l_i$ is an epicenter?,
\end{itemize}
where $C^T=(c_1^T,c_2^T,\dots,c_d^T)$ and $c^T_i(x)$ judges whether trajectory $x$ includes location $l_i$ in time $T$ or not.

\section{Insufficiency of Existing Privacy Definitions}
\label{sec:sufficiency}
Here, we clarify the insufficiencies of the existing privacy definitions: $\varepsilon$-DP, $(\varepsilon,\delta)$-DP, Blowfish privacy, and OSDP, which is our motivation to introduce ADP for our problem.
Concretely, we show that unreasonable privacy protection is required for a mechanism to be $(q^i_\leq,\alpha,\beta)$-one-sided for each privacy definition.
To do so, we first derive the necessary condition for a mechanism to be $(q^i_\leq,\alpha,\beta)$-one-sided.
\begin{proposition}
\label{prop:necessary condition for one-sided error}
If $M$ is $(q,\alpha,\beta)$-one-sided for $X$ such that $q(X)={\rm True}$,
$\forall X^\prime$ such that $q(X^\prime)={\rm False}$, $\exists S\subseteq \mathcal{Z}$
\begin{equation}
\label{eq:necessarcy condition}
\alpha\Pr[M(X)\in S] > \beta \Pr[M(X^\prime)\in S].
\end{equation}
\end{proposition}


Using this necessary condition, we clarify the insufficinecies in the case $q^i_\leq =p^t_\leq \circ f_{C^T}(\cdot)_i$. 
Note that the same is true in the case of $q^i_\geq =p^t_\leq \circ f_{C^T}(\cdot)_i$.

\subsection{$\varepsilon$-DP}
There exist two datasets $X,X^\prime$ such that $q(X)={\rm True}$, $q(X^\prime)={\rm False}$, and $d_{\rm ham}(X,X^\prime)=k<\infty$, where $d_{\rm ham}$ is the hamming distance.
Therefore, if mechanism $M$ satisfies $\varepsilon$-DP, $\forall S\subseteq \mathcal{Z},$
\begin{equation}
\label{eq:dp_condition}
\Pr[M(X)\in S] \leq \mathrm{e}^{k\varepsilon} \Pr[M(X^\prime)\in S].
\end{equation}
If $M$ achieves both of $(q,\alpha,\beta)$-one-sided and $\varepsilon$-DP, by combining two conditions $(\ref{eq:necessarcy condition})$ and $(\ref{eq:dp_condition})$,
\begin{equation}
\label{eq:epsilon dp condition}    
\varepsilon > \frac{\log (\beta/\alpha)}{k}.
\end{equation}

For example, when $k=1, \alpha=10^{-4}$, and $\beta=1/2$, it must hold that $\varepsilon > 8.4$ to achieve $(q,\alpha,\beta)$-one-sided.
Such value of $\varepsilon$ is abnormal~\cite{hsu2014differential} because $\varepsilon$ is the exponent of $\mathrm{e}$.
Moreover, when $\alpha=0$, there does not exists $\varepsilon<\infty$.

\subsection{$(\varepsilon,\delta)$-DP}
\label{sec:insufficiency_dp}
We can show the insufficiency by the same way as the case of $\varepsilon$-DP.
If a mechanism $M$ satisfies $(\varepsilon,\delta)$-DP, $\forall S\subseteq\mathcal{Z}$,
\begin{equation}
\label{eq:approximate dp condition}
    \Pr[M(X)\in S]\leq \mathrm{e}^{k\varepsilon}\Pr[M(X^\prime)\in S]+\delta\sum_{i=1}^k \mathrm{e}^{(i-1)\varepsilon}.
\end{equation}
If $M$ achieves both of $(q,\alpha,\beta)$-one-sided and $(\varepsilon,\delta)$-DP, by combining two conditions $(\ref{eq:necessarcy condition})$ and $(\ref{eq:approximate dp condition})$,
\begin{equation}
\label{ineq:approximate dp condition}
\delta > \frac{\beta-\alpha\mathrm{e}^{\varepsilon k}}{\sum_{i=1}^{k} \mathrm{e}^{(i-1)\varepsilon}}.
\end{equation}
For example, when $k=3, \alpha=0, \beta=1/2$, and $\varepsilon=1$, it must hold that $\delta>0.12$.
This value of $\delta$ is clearly unreasonable because $\delta$ is the probability that $\varepsilon$-DP breaks.

\subsection{Blowfish privacy and OSDP}
\label{sec:insufficiency_adversary}
Here, we assume that $\varepsilon<\frac{\log (\alpha/\beta)}{k}$, so an $\varepsilon$-DP mechanism fails to be $(q,\alpha,\beta)$-one-sided from the condition $(\ref{eq:epsilon dp condition})$.
Hence, we consider to relax $\varepsilon$-DP by the ways of Blowfish privacy (policy graph) and OSDP (policy function), respectively so that a mechanism is $(q,\alpha,\beta)$-one-sided.
Then, we evaluate privacy protection of them; however, how to evaluate them is not trivial because policy graph and policy function are qualitative.
To this end, we first introduce a quantitative function based on Bayesian analysis~\cite{mironov2017renyi}, called Risk, which shows the insufficiencies.

\paragraph*{Risk}
Given dataset $X=(x^1,x^2,\dots,x^n)$, we consider the privacy of $x^1$ without loss of generality.
We consider $x^1$ as a realization of random variable $x^\pi$ which follows probability distribution $\pi$ from the view of an adversary with knowledge $\pi$.
We assume that the adversary tries to infer whether $x^1$ belongs to an attribute designated by $p$ (i.e., whether $p(x^1)={\rm True}$ or ${\rm False}$).
For $b\in \{\rm False, True\}$, we model the belief about $p$ of the adversary before seeing an output by
$$
\Pr[p(x^\pi)=b]:=\sum_{x\in \{x|p(x)=b\}}\pi (x)
$$
and the belief after seeing output $M(X)=z$ by
$$
\Pr[p(x^\pi)=b|z]:=\sum_{x\in\{x|p(x)=b\}}\pi (x|z),
$$
where $\pi(\cdot|z)$ is the posterior distribution given $z$.
Then, we define {\rm Risk} as the maximum odds ratio; ${\rm Risk}_{\pi}$ measures how much the adversary updates the belief.
${\rm Risk}_{\pi}(M,p):=$
$$
\max_{z\in \mathcal{Z}} \left(\left.
\frac{\Pr[p(x^\pi)={\rm True}|z]}{\Pr[p(x^\pi)={\rm False}|z]}\middle/ \frac{\Pr[p(x^\pi)={\rm True}]}{\Pr[p(x^\pi)={\rm False}]}\right. \right).
$$
When $\Pr [p(x^\pi)={\rm False}]=0$, we define ${\rm Risk}_{\pi}(M,p)$ as $0$ because the adversary originally knows the information $p(x^1)={\rm True}$ and nothing updates.
Otherwise, we follow the convention $a/0:=\infty$ for $a>0$.
Intuitively, if the adversary more gets the information $p(x^1)={\rm True}$ from $z$, Risk becomes greater.
If the adversary may exactly get the information from $z$, Risk is $\infty$.
For example, if mechanism $M$ satisfies $\varepsilon$-DP, for any $p$ and $\pi$,
${\rm Risk}_{\pi}(M,p)\leq \mathrm{e}^\varepsilon.$
That is, $\varepsilon$-DP protects any information against any adversary in this setting.
\iflong
\begin{proof}
Consider random variable $X=(x^\pi,x^2,\dots,x^n)$ where $x^2,x^3,\cdots,x^n$ are constant.
Then, for $b\in\{\rm True, False\}$, $\Pr[M(X)=z|p(x^\pi)=b]$ is
$$
\sum_{\{x|p(x)\}}\frac{\Pr[x^\pi=x]}{\Pr[p(x^\pi)=b]}\Pr[M(X)=z|x^\pi=x].
$$
Therefore, by applying the definition of DP and adjusting the coefficient for each term, if mechanism $M$ satisfies $\varepsilon$-DP,
$$
\Pr [M(X)=z|p(x^\pi)={\rm True}]\leq \mathrm{e}^\varepsilon \Pr[M(X)=z|p(x^\pi)={\rm False}],
$$
which derives what we wanted.
\end{proof}
\fi
\subsubsection{Blowfish privacy}
\label{sec:insufficiency_blowfish}
In Blowfish privacy, the neighboring relationship is stipulated by graph $G=(\mathcal{X},E)$.
We construct a graph which allows a mechanism that is $(q,\alpha,\beta)$-one-sided from the necessary condition~(\ref{eq:necessarcy condition}).
We can derive the following necessary condition for a policy graph from the definitions.
\begin{proposition}
Given $\varepsilon<\frac{\log (\alpha/\beta)}{k}$, if mechanism $M$ is $(p^t_\leq \circ f_{C^T}(\cdot)_i,\alpha,\beta)$-one-sided and satisfies $(\varepsilon,(\mathcal{X},E))$-Blowfish privacy,
$\forall x,x^\prime\in\mathcal{X}$ such that $c^T_i(x)\neq c^T_i(x^\prime)$,
\begin{equation}
\label{eq:necesary_condition_graph}
(x,x^\prime) \notin E.
\end{equation}
\end{proposition}
\iflong
\begin{proof}
We let $q=p^t_\leq \circ f_{C^T}(\cdot)_i$ for simplicity.
Assume the following which derives the contradiction: $\exists x,x^\prime$ such that $c_i^T(x)\neq c^T_i(x^\prime)$, $(x,x^\prime)\in E$.
In this case, there exists $X,X^\prime$ such that $X\sim_{x,x^\prime} X^\prime$ and $q(X)={\rm True}$ and $q(X^\prime)={\rm False}$.
Since $(x,x^\prime)\in E$,
$$
\Pr[M(X)=z]\leq\mathrm{e}^\varepsilon\Pr[M(X^\prime)=z]
$$
from the definition of Blowfish privacy.
Combining this to (\ref{eq:necessarcy condition}) derives $\varepsilon>\log(\alpha/\beta)$, which is contradiction.
\end{proof}
\fi
Then, we can show that there exists $(\varepsilon,G)$-blowfish private mechanism $M$ such that $\exists \pi$,
$
{\rm Risk}_\pi (M,\neg c^T_i)=\infty\ {\rm and}\ {\rm Risk}_\pi (M,c^T_i)=\infty,
$
when $G$ satisfies the condition~(\ref{eq:necesary_condition_graph}).
This implies that an adversary gets to know the information that $\neg c^T_i(x^1)={\rm True}$ or $c^T_i(x^1)={\rm True}$ from the output of a mechanism.
When $\neg c^T_i(x^1)={\rm True}(x)$ or $c^T_i(x^1)={\rm True}(x)$ is sensitive, (it is often that either of them is sensitive as described in Introduction), this privacy protection is unreasonable.

\subsubsection{OSDP}
\label{sec:insufficiency_osdp}
In OSDP, the neighboring relationship is stipulated by function $P_N:\mathcal{X}\to \{\rm False, True\}$, which judges whether the input record is non-sensitive or not.
We construct a function which allows a mechanism that is $(q,\alpha,\beta)$-one-sided from the necessary condition~(\ref{eq:necessarcy condition}).
We can derive the following necessary condition for a policy function from the definitions.
\begin{proposition}
Given $\varepsilon<\frac{\log (\alpha/\beta)}{k}$, if a mechanism $M$ is $(p^t_\leq\circ f_{C^T}(\cdot)_i,\alpha,\beta)$-one-sided and satisfies $(\varepsilon, P_N)$-OSDP,
$\forall x\in\mathcal{X}$ such that $c^T_i(x)={\rm True}$,
\begin{equation}
\label{eq:necessary_condition_policy_function}
x\in N.
\end{equation}
\end{proposition}
\iflong
\begin{proof}
We let $q=p^t_\leq \circ f_{C^T}(\cdot)_i$ for simplicity.
Assume the following which derives the contradiction: $\exists x$ such that $c_i^T(x)={\rm True}$, $x\notin N$.
In this case, there exists $X,X^\prime$ such that $X\sim_{x,x^\prime} X^\prime$ and $q(X)={\rm True}$ and $q(X^\prime)={\rm False}$.
Since $P_N(x)={\rm True}$,
$$
\Pr[M(X)=z]\leq\mathrm{e}^\varepsilon\Pr[M(X^\prime)=z]
$$
from the definition of OSDP.
Combining this to (\ref{eq:necessarcy condition}) derives $\varepsilon>\log(\alpha/\beta)$, which is contradiction.
\end{proof}
\fi
Then, we can show that there exists $(\varepsilon,P_N)$-OSDP mechanism such that $\exists \pi$,
$
{\rm Risk}_\pi (M,(\neg c^T_{i})\land p)=\infty,
$
where $N$ satisfies the condition~(\ref{eq:necessary_condition_policy_function}) and $(\neg c^T_{i})\land p(x):=(\neg c^T_i(x))\land p(x)$ where $p$ is any propositional function to $x\in \mathcal{X}$.
This implies that an adversary may get to know any information of $x$ if the owner did not visit location $i$ at time $T$.
This is clearly unreasonable.

\paragraph*{Remark}
If $\mathcal{X}=\{\rm False, True\}$ that indicates whether the user visited or not the target location instead of trajectories, $(\varepsilon, P_N)$-OSDP is reasonable, where $N=\{\rm False\}$, under the assumption that the information that a user did not visit the target location is non-sensitive because $x\in\mathcal{X}$ does not include any other information.
Therefore, if the target is one location, OSDP can provide reasonable privacy protection.
Note that sequentially applying this for multiple locations return to the above problem.

\section{Asymmetric Differential Privacy}
\label{sec:propose}
As shown in the previous section, existing privacy definitions result in unreasonable privacy protection to achieve one-sided error.
To solve this problem, we propose ADP, which provides reasonable privacy protection while achieving one-sided error.
ADP is also the definition that changes the neighboring relationship of $\varepsilon$-DP, so ADP belongs to the same family as Blowfish privacy and OSDP.

First, we describe the definition of ADP.
Second, we state the privacy guarantee of $(\varepsilon,p)$-ADP.
Third, we introduce $p^i_\leq$ and $p^i_\geq$, which are policies to answer $q^i_\leq$ and $q^i_\geq$ (see Section~\ref{subsec:problem_setting} for them) with one-sided error.

\subsection{Definition}
\label{subsec:definition}
First, we define new neighboring relationship, called $p$-neighboring relationship, using propositional function $p:\mathcal{X}\to \{\rm False, True\}$ called \textit{policy}.
Formally,
\begin{definition}[$p$-neighboring relationship]
Given two datasets $X$ and $X^\prime$,
$X$ is $p$-neighboring to $X^\prime$, if $X\sim_{x,x^\prime}X^\prime$ and ($\neg p(x)$ or ($p(x)\land p(x^\prime)$))$={\rm False}$ (i.e., $(p(x)\to p(x^\prime))={\rm False}$).
We denote this by $X\sim^p_{x,x^\prime} X^\prime$ (or simply $X\sim^p X^\prime$).
\end{definition}
Note that $X\sim^p X^\prime$ does not mean $X^\prime\sim^p X$, which is the origin of the name of \textit{asymmetric} differential privacy.
Based on $p$-neighboring relationship, we define ADP.
\begin{definition}[ADP]
Randomized mechanism $M$ satisfies $(\varepsilon,p)$-ADP iff $\forall z\in\mathcal{Z}, X,X^\prime$ such that $X\sim^p X^\prime$,
$$
\Pr[M(X)=z]\leq\mathrm{e}^\varepsilon \Pr[M(X^\prime)=z].
$$
\end{definition}

\paragraph*{Remark}
We design $p$-neighboring relationship to circumvent the conflict of conditions (\ref{eq:necessarcy condition}) and (\ref{eq:dp_condition}) so that a mechanism can be $(q,\alpha,\beta)$-one-sided for $\varepsilon<\frac{\log (\alpha/\beta)}{k}$.
Therefore, ADP is the natural relaxation to achieve one-sided error.

\subsection{Privacy guarantee of $(\varepsilon,p)$-ADP}
\label{subsec:privacy guarantee}
In the previous section, we define $(\varepsilon,p)$-ADP to achieve one-sided error by relaxing DP, so it is not clear that the privacy guarantee of $(\varepsilon,p)$-ADP is reasonable.
Hence, we describe the privacy guarantee of $(\varepsilon,p)$-ADP.
The relaxation by $p$ is qualitative just like Blowfish privacy and OSDP, so we use Risk introduced in Section~\ref{sec:insufficiency_adversary} to evaluate the privacy of $(\varepsilon,p)$-ADP.
This analysis shows what information is leaked and protected by mechanism $M$ which satisfies $(\varepsilon,p)$-ADP.

\subsubsection{Leaked information}
\label{sec:leaked information}
Given two datasets $X\sim_{x,x^\prime} X^\prime$, $X$ is not $p$-neighboring to $X^\prime$ if $p(x)={\rm True}$ and $p(x^\prime)={\rm False}$.
This implicates that $M(X)$ is distinguishable from $M(X^\prime)$ for such $X$ and $X^\prime$ even if $M$ satisfies $(\varepsilon,p)$-ADP.
Therefore, an adversary may infer the information $p(x)={\rm True}$ from output of $M(X)$.
Formally,
\begin{proposition}
\label{prop:leaked_information}
There exists $(\varepsilon,p)$-ADP mechanism $M$ such that $\exists \pi,$
$
{\rm Risk}_{\pi}(M,p) = \infty.
$
\end{proposition}
\iflong
\begin{proof}
Consider the following mechanism $M$:
For $X$ such that $|\{x\in X|p(x)={\rm False}\}|>0$,
\begin{align}
\nonumber
M(X)=
\left\{ \begin{array}{ll}
{\rm True} & ({\rm w.p.}\ 0) \\
{\rm False} & ({\rm w.p.}\ 1) \\
\end{array} \right.
\end{align}
For $X$ such that $|\{x\in X|p(x)={\rm False}\}|=0$,
\begin{align}
\nonumber
M(X)=
\left\{ \begin{array}{ll}
{\rm True} & ({\rm w.p.}\ 1-1/\mathrm{e}^{\varepsilon}) \\
{\rm False} & ({\rm w.p.}\ 1/\mathrm{e}^{\varepsilon}) \\
\end{array} \right.
\end{align}
This mechanism satisfies $(\varepsilon,p)$-ADP and $\exists \pi$, ${\rm Risk}_{\pi}(M,p) = \infty$.
\end{proof}
\fi

\subsubsection{Protected information}
First, we clarify information $(\varepsilon,p)$-ADP protects against a general adversary (i.e., any prior knowledge $\pi$).
\begin{proposition}
\label{prop:protection_not_relate}
If mechanism $M$ satisfies $(\varepsilon,p)$-ADP, for any $\pi$ and $p^\prime:\mathcal{X}\to\{{\rm False, True}\}$,
$
{\rm Risk}_{\pi}(M,(\neg p)\land p^\prime) \leq  \mathrm{e}^\varepsilon,
$
where $(\neg p)\land p^\prime (x):=(\neg p(x))\land p^\prime (x)$.
\end{proposition}
This proposition is derived from the definition of $(\varepsilon,p)$-ADP.
\iflong
\begin{proof}
Here, we let $p^{\prime\prime}=(\neg p)\land p^\prime$ for simplicity.
From the definition of $(\varepsilon,p)$-ADP, $\forall z\in \mathcal{Z}, x^\prime\in\mathcal{X}$ and $X\sim_{x,x^\prime} X^\prime$ such that $p(x)={\rm False}$,
$$
\Pr[M(X)=z]\leq \mathrm{e}^\varepsilon \Pr[M(X^\prime)=z].
$$
From the definition of $p^{\prime\prime}$, when $p^{\prime\prime}(x)={\rm True}$, $p(x)={\rm False}$.
Therefore, 
\begin{equation}
\nonumber
\begin{split}
&\Pr[M(X)=z|p^{\prime\prime}(x^1)={\rm True}]\leq\\
&\mathrm{e}^\varepsilon\Pr[M(X)=z|p^{\prime\prime}(x^1)={\rm False}].
\end{split}
\end{equation}
which induces ${\rm Risk}_{\pi}(M,(\neg p)\land p^\prime) \leq  \mathrm{e}^\varepsilon$.
\end{proof}
\fi
This proposition implicates that $(\varepsilon,p)$-ADP mechanism protects information that does not relate to the information $p(x^1)={\rm True}$.
In other words, if $p(x^1)={\rm False}$, $(\varepsilon,p)$-ADP protects any information about $x^1$.
One may wonder if some leakage of information except $p(x^1)={\rm True}$ occurs when $p(x^1)={\rm True}$.
Surprisingly, we can show that $(\varepsilon,p)$-ADP protects any information except $p(x^1)={\rm True}$.

Here, we explore whether any information leaks or not in addition to the information $p(x^1)={\rm True}$.
To do so, we assume that $p(x^1)={\rm True}$ and an adversary knows the information.
Then, we explore how the adversary updates the knowledge from output $M(X)$.
We model the adversary who knows $p(x^1)={\rm True}$ by $\pi^p$, where ${\rm Supp}(\pi^p)=\{x\in\mathcal{X}|p(x)\}$.
Here, ${\rm Supp}(\pi^p)$ is the support of $\pi^p$.
Then, we can derive the following proposition from the definition of $(\varepsilon,p)$-ADP.
\begin{proposition}
\label{prop:protection_relate}
If $M$ satisfies $(\varepsilon,p)$-ADP, for any $p^\prime:\mathcal{X}\to{\{\rm True, False\}}$ and $\pi^p$ such that ${\rm Supp}(\pi^p)=\{x\in\mathcal{X}|p(x)\}$,
$
{\rm Risk}_{\pi^p}(M,p^\prime) \leq  \mathrm{e}^\varepsilon.
$
\end{proposition}
\iflong
\begin{proof}
From the definition of $(\varepsilon,p)$-ADP, for any $z\in \mathcal{Z}$ and $X\sim_{x,x^\prime} X^\prime$ such that $p(x)={\rm True}$ and $p(x^\prime)={\rm True}$,
$$
\Pr[M(X)=z]\leq \mathrm{e}^\varepsilon \Pr[M(X^\prime)=z].
$$
Therefore, for any $p^\prime$,
\begin{equation}
\nonumber
\begin{split}
&\Pr[M(X)=z|p(x^1)\land p^\prime(x^1)={\rm True}]\leq\\ &\mathrm{e}^\varepsilon \Pr[M(X)=z|p(x^1)\land \neg p^\prime(x^1)={\rm True}].
\end{split}
\end{equation}
Therefore, ${\rm Risk}_{\pi^p}(M,p^\prime) \leq  \mathrm{e}^\varepsilon$.
\end{proof}
\fi

\paragraph*{Remark}
These propositions imply the relationship between ADP and OSDP.
From Proposition~\ref{prop:protection_not_relate}, we can soon derive the following inequality.
\begin{equation}
\label{eq:generalization_of_sensitivity_masking}
{\rm Risk}_{\pi}(M,\neg p) \leq \mathrm{e}^\varepsilon
\end{equation}
Given a set of non-sensitive records $N$, consider $(\varepsilon, p=P_N)$-ADP (refer to Section~\ref{subsec:generalization of dp} for $P_N$).
From the implication of Inequality~(\ref{eq:generalization_of_sensitivity_masking}), $(\varepsilon, p)$-ADP protects whether $x$ is non-sensitive or not.
This matches \textit{sensitivity masking} introduced by Doudalis et al.~\cite{doudalis2020one}.
That is, Proposition~\ref{prop:protection_not_relate} is generalization of sensitivity masking.
It is worth noting that OSDP does not induce Proposition~\ref{prop:protection_relate}.
These two differences come from the fact that ADP discriminates information a record may have into sensitive and non-sensitive while OSDP discriminates a record itself into sensitive and non-sensitive.
That is, ADP is more fine-grained privacy relaxation than OSDP, which prevents the insufficiency described in Section~\ref{sec:insufficiency_osdp}.

\subsection{Policies}
Here, we introduce polices to achieve one-sided error for $q^i_{\leq}$ and $q^i_{\geq}$, respectively.
Cutting it short, setting $\neg c^T_i$ and $c^T_i$ allows one-sided error for $q^i_{\leq}$ and $q^i_{\geq}$, respectively.

\subsubsection{Policy for $q^i_{\leq}$}
\label{sec:one-sided error for q_leq}
Here, we consider two datasets $X$ and $X^\prime$ such that $q^i_\leq(X)={\rm True}$ and $q^i_\leq(X^\prime)={\rm False}$.
The conditions of (\ref{eq:necessarcy condition}) and (\ref{eq:dp_condition}) conflict because there is a path between $X$ and $X^\prime$ via the neighboring relationship\footnote{"There is a path between $X$ and $X^\prime$ via a neighboring relationship (e.g., $\sim$)" means that there is a sequence of datasets $(X,X^1,\dots,X^l,X^\prime)$ such that $X\sim X^1$, $X^1\sim X^2,\dots, X^l\sim X^\prime$.}.
Therefore, setting $p$ so that there is no path between $X$ and $X^\prime$ via $p$-neighboring avoids the conflict.
To stipulate such neighboring relationship, we introduce $p$ such that $\forall x\in\mathcal{X},$
$
\neg c^T_i(x)\to p(x) = {\rm True}
$

Then, using such $p$, we analyze one-sided error of an $(\varepsilon,p)$-ADP mechanism.
To this end, we first introduce the minimum path instead of the hamming distance in DP.
\begin{definition}[the minimum path]
The minimum path from $X$ to $X^\prime$, denoted by $d^p_{\rm min}(X,X^\prime)$, is the minimum number of steps to take to change $X$ to $X^\prime$ via $p$-neighboring datasets.
\end{definition}
If $d^p_{\rm min}(X,X^\prime)=k$, there is a $k$ steps path.
Using the minimum path, we can analyze the one-sided error of an $(\varepsilon,p)$-ADP mechanism as follows.
\begin{proposition}
\label{prop:adp_onesided}
If $(\varepsilon,p)$-ADP mechanism is $(q,\alpha,\beta)$-one-sided for $X$,
$
\varepsilon \geq \frac{\log((1-\alpha)/(1-\beta))}{k(X)},
$
where $k(X)=\min_{\{X^\prime|q(X^\prime)={\rm False}\}} \{d^p_{\rm min}(X,X^\prime)\}$.
\end{proposition}
Important remark is that $\varepsilon$ can be small value even when $\alpha=0$ by relaxation of $p$.
This relaxation of $p$ leads to the leakage of the information that the user did not visit the location $l_i$.
We believe that the information is relatively non-sensitive.
Therefore, we adopt this policy.

\subsubsection{Policy for $q^i_{\geq}$}
\label{sec:one-sided error for q_geq}
In the similar vein, for a mechanism to be $(q^i_{\geq},\alpha,\beta)$-one-sided and achieve $(\varepsilon,p)$-ADP, we introduce $p$ such that $\forall x\in\mathcal{X},$
$
c^T_i(x)\to p(x) ={\rm True}.
$
From proposition~\ref{prop:leaked_information},
this policy leaks information that the user visited location $l_i$.
Therefore, we need to assume that such information is non-sensitive; however we believe such assumption is not realistic.
Hence, we do not adopt this policy, so our mechanism will provide the same performance for $q^i_{\geq}$ as $\varepsilon$-DP.

\section{Mechanisms}
\label{sec:asvt}
Here, we propose mechanisms that are $(q^i_\leq,\alpha,\beta)$-one-sided and satisfy $(\varepsilon,p)$-ADP.

\subsection{Asymmetric Laplace mechanism}
\label{sec:adp_primitive_mechanism}
We introduce an asymmetric Laplace mechanism (ALap), the ADP version of the Laplace mechanism, as the most basic mechanism.
ALap perturbs the answer of counting query $f:\mathcal{X}^n\to\mathbb{R}^d$.
First, we define $p$-sensitivity, which corresponds to the notion of \textit{sensitivity} of DP.
Then, we define ALap using $p$-sensitivity\footnote{The ADP version of the geometric mechanism~\cite{ghosh2012universally} is optimal and better than the Laplace mechanism for counting query, but we omit it due to space limitations and simplicity, and we refer to the full version for the details.}.

\subsubsection{$p$-sensitivity}
We define $p$-sensitivity of query $f:\mathcal{X}^n\to\mathbb{R}^d$ as the sensitivity induced by policy $p$.
\begin{definition}[$p$-sensitivity of $f$]
Given $p:\mathcal{X}\to\{{\rm True, False}\}$, we define $p$-sensitivity of $f$, denoted by $\Delta_{p}$, as follows:
$
    \Delta_{p}(f) := \sup_{X\sim^p X^\prime} ||f(X^\prime)-f(X) ||_1 .
$
\end{definition}
$p$-sensitivity has a unique characteristic called \textit{monotonicity} which does not appear in the sensitivity of DP.
\begin{definition}[Monotonicity of $p$-sensitivity]
The $p$-sensitivity of $f(\cdot)_i$ ($\Delta_{p}(f(\cdot)_i)$) is monotonically increasing (decreasing) iff $\forall X\sim^p X^\prime$, 
$
f(X)_i\leq f(X^\prime)_i\  (f(X)_i\geq f(X^\prime)_i).
$
If for all $i\in [d]$, $\Delta_{p}(f(\cdot)_i)$ is monotonically increasing (decreasing), we simply say that the $p$-sensitivity of $f$ is monotonically increasing (decreasing).
Also, we define function ${\rm Sign}_p$, which discriminates the monotonicity. ${\rm Sign}_p(f(\cdot)_i):=$
$$
\begin{cases}
+1\ {\rm(if\ } \Delta_{p}(f(\cdot)_i){\rm\ is\ monotonically\ increasing)}\\
-1\ {\rm(if\ } \Delta_{p}(f(\cdot)_i){\rm\ is\ monotonically\ decreasing)}
\end{cases}
$$
\end{definition}

\subsubsection{Definition}
\label{sec:alap}
We propose ALap, which perturbs the answer of query $f:\mathcal{X}^n\to\mathbb{R}^d$ according to $p$-sensitivity\footnote{This mechanism is a generalization of OSDPLaplace proposed by Doudalis et al.~\cite{doudalis2020one}.}.

\begin{definition}[Asymmetric Laplace mechanism]
Given query $f:\mathcal{X}^n\to\mathbb{R}^d$, policy $p$, and privacy parameter $\varepsilon$, the asymmetric Laplace mechanism $\rm{ALap}$ is as follows:
$
    {\rm ALap}_{p,\varepsilon, f}(X)=f(X)+(\lambda_1,\dots,\lambda_d),
$
where $\lambda_i$ is independently distributed and follows the distribution for each $i\in[d]$:\\
if $p$-sensitivity of $f(\cdot)_i$ is monotonic, 
\begin{align}
\nonumber
\left\{ \begin{array}{ll}
\frac{\varepsilon}{\Delta_{p}(f)}\exp{\frac{-{\rm Sign}_p(f(\cdot)_i)\lambda\varepsilon}{\Delta_{p}(f)}}& ({\rm Sign}_p(f(\cdot)_i)\lambda\geq 0) \\
0 & ({\rm otherwise}). \\
\end{array} \right.
\end{align}
Otherwise,
$
\frac{\varepsilon}{2\Delta_{p}(f)}\exp{\frac{\left|\lambda\right|\varepsilon}{\Delta_{p}(f)}}.
$
\end{definition}
If the $\Delta_{p}(f(\cdot)_i)$ is not monotonic, ALap is the same as the Laplace mechanism for $i$th output.
Otherwise, the distribution is the (one-sided) exponential distribution, which has a smaller variance than the Laplace distribution.

\begin{theorem} 
\label{theo:adp_laplace_privacy}
Given $p$, $\varepsilon\in\mathbb{R}^+, f:\mathcal{X}^n\to\mathbb{R}^l$, ALap$_{p,\varepsilon,f}$ satisfies ($\varepsilon, p$)-ADP\@.
\end{theorem}

\iflong
\begin{proof}
Let $M$ be ALap$_{p,\varepsilon, f}$, and assume arbitrary datasets $X$ and $X^\prime$ such that $X\sim^p X^\prime$.
We have
\begin{align}
\nonumber\frac{\Pr[M(X)=(z_1,\dots,z_d)]}{\Pr[M(X^\prime)=(z_1,\dots,z_d)]}&=\prod_{i\in [d]}\frac{\Pr[\lambda_i=z_i-f(X)_i]}{\Pr[\lambda_i=z_i-f(X^\prime)_i]}\\
&\nonumber\leq\exp{(\varepsilon\left|f(X^\prime)-f(X)\right|/\Delta_{p}(f))}\\
&\nonumber\leq\exp{(\varepsilon)}.
\end{align}
The first inequality is from the definition of the distribution.
\end{proof}
\fi

\subsubsection{Utility analysis}
\label{subsubsec:alap_utility_analysis}
Assume that the $\Delta_{p}(f)$ is monotonically decreasing.
In this case, ALap$_{p,\varepsilon,f}$ uses noise following the (one-sided) exponential distribution.
From this characteristic, we easily derive the following two corollaries about utility:
\begin{corollary}
For all $i\in [d]$,
$
\mathbb{E}[|{\rm ALap}_{p,\varepsilon,f}(X)_i-f(X)_i|]=\frac{\Delta_{p}(f)}{\varepsilon},
$
where the randomness is over the mechanism.
\end{corollary}
\begin{corollary}
\label{col:utility one-sided alap}
Given dataset $X$, if $\varepsilon\geq\frac{\Delta_{p}(f)\log(1/(1-\beta))}{k(X)}$,
ALap$_{p,\varepsilon,f}$ is $(p^t_\leq \circ f_{C^T}(\cdot)_i,\alpha,\beta)$-one-sided for all $i\in [d]$.  
\end{corollary}
This is because the output True is accurate by using $p^t_\leq$ as $p^{\rm san}$ (i.e., outputting $p^t_\leq\circ {\rm ALap}_{p,\varepsilon,f}$).

We can see that ALap improves the utility of the Laplace mechanism, which is two-sided and whose expected absolute error is $\sqrt{2}\Delta/\varepsilon$.
Especially, Corollary~\ref{col:utility one-sided alap} does not depend on $\alpha$.
Note that ${\rm ALap}_{p,\varepsilon,f}$ is optimal with respect to one-sided error when $\Delta_{p}(f)=1$ from Proposition~\ref{prop:adp_onesided}.

\subsection{Sanitized Asymmetric Laplace mechanism}
As shown in Corollary~\ref{col:utility one-sided alap}, ALap requires large $\varepsilon$ proportional to $\Delta_{p}(f)$.
To solve this problem, we propose the sanitized asymmetric Laplace mechanism (SALap), which is robust to large $\Delta_{p}(f)$.

\subsubsection{Definition}
SALap handles $f$, where $\Delta_p(f(\cdot)_i)$ is monotonic and $\Delta_p(f(\cdot)_i)\leq 1$ for all $i\in [d]$.
SALap removes the dependency of $\Delta_{p}(f)$ by sequentially answering $f(\cdot)_i$ from $i=1$.
Informally, outputting ALap$_{p,\varepsilon,f(\cdot)_i}$ consumes $\varepsilon$, but the post-processing output $p^t_\leq \circ $ALap$_{p,\varepsilon,f(\cdot)_i}$ does not consume any privacy budget when it returns True and only consumes when it returns False.
Therefore, we can continue to answer the sequence of queries until outputting False.
Utilizing this characteristic, we construct SALap as Algorithm~\ref{alg:asvt}.
SALap sequentially computes $z=$ALap$_{p,\varepsilon,f(\cdot)_i}$ from $i=1$, and if $p^t_\leq(z)$ is True, it outputs $\perp$; otherwise, it outputs $z$.
SALap stops when SALap outputs a value other than $\perp$.

SALap can answer each query using ALap with the parameter $\varepsilon$ instead of $\varepsilon/\Delta_{p}(f)$.
Therefore, SALap answers the sequence with one-sided error which matches the optimal performance of Proposition~\ref{prop:adp_onesided}.
However, when the sequence includes many queries whose answer is above the threshold, it quickly stops due to the abort operator.
That is, SALap works well for a sequence of queries whose answers are sparse.

\paragraph*{Remark}
SALap is very similar to the sparse vector technique (SVT) of DP~\cite{lyu2016understanding}.
However, ADP improves SVT from the following four differences.
First, SALap does not require noise for the threshold $t$.
Second, we can use parameter $\varepsilon$ for ALap instead of $\varepsilon/2$.
Third, we can answer a noisy counting if the sanitized output is False.
Fourth, the output is one-sided (i.e., $\perp$).

 \begin{algorithm}[t]
 \caption{Sanitized Asymmetric Laplace mechanism \label{alg:asvt}}
 \begin{algorithmic}[1]
 \renewcommand{\algorithmicrequire}{\textbf{Input:}}
 \renewcommand{\algorithmicensure}{\textbf{Output:}}
 \Require $X:$ dataset, $p:$ policy, $\varepsilon:$ privacy parameter, $\{f_1, f_2, \dots, f_d\}$: a sequence of counting queries whose $p$-sensitivity is monotonically decreasing and equal or less than $1$, $\{p^{t_1}_\leq, p^{t_2}_\leq, \dots, p^{t_k}_\leq\}$: a set of threshold propositions
 \State aborted $\Leftarrow\ {\rm False}$
  \For {$i\in [d]$}
  \State $z\Leftarrow {\rm ALap}_{p,\varepsilon,f_i}(X)$
  \If {aborted $=$ True}
    \State $a_i={\rm nan}$
    \State pass
\Else
  \If {$p^{t_i}_\leq(z)$}
  \State Output $a_i=\perp$
  \Else
  \State Output $a_i=z$
  \State aborted $\Leftarrow\ {\rm True}$
  \EndIf
  \EndIf
  \EndFor
 \Ensure  $A=(a_1,a_2,\dots,a_d)$
 \end{algorithmic} 
 \end{algorithm}

\subsubsection{Privacy Analysis}
Here, we show that SALap satisfies $(\varepsilon,p)$-ADP.
 \begin{theorem}
 Algorithm~\ref{alg:asvt} satisfies ($\varepsilon,p$)-ADP\@.
 \end{theorem}

\iflong
\begin{proof}
We separately consider two cases: output $\perp$ (line 9) and output $z$ (line 11).
We let $M_i$ denote the mechanism at round $i$ (i.e., lines 2 to 12).
First, we prove that when $M_i$ outputs $\perp$ (i.e., $a_i=\perp$), it holds that for $b\in \{{\rm True, False}\}$,
\begin{equation}
\label{eq:asvt_-1}
\frac{\Pr[M_i(X)=\perp|{\rm aborted}=b]}{\Pr[M_i(X^\prime)=\perp|{\rm aborted}= b]}\leq 1 ,
\end{equation}
where $X\sim ^p X^\prime$.
Since $f_i$ is monotonically decreasing, when $b={\rm False}$,
\begin{align}
\nonumber\Pr[M_i(X)=\perp|b]&=\Pr[{\rm ALap}_{p,\varepsilon,f_i}(X)-t_i\leq 0]\\
\nonumber&=\Pr[\lambda\leq t_i-f_i(X)]\\
\nonumber&\leq\Pr[\lambda\leq t_i-f_i(X^\prime)]\\
\nonumber&=\Pr[{\rm ALap}_{p,\varepsilon,f_i}(X^\prime)-t_i\leq 0]\\
\nonumber&=\Pr[M_i(X^\prime)=\perp|b],
\end{align}
where $\lambda$ is the distribution of ALap.
Therefore, inequality (\ref{eq:asvt_-1}) holds.
The second case is proved from Theorem~\ref{theo:adp_laplace_privacy}: for $b\in\{{\rm True, False}\}$,
$$\frac{\Pr[M_i(X)=z|b]}{\Pr[M_i(X^\prime)=z|b]}\leq \varepsilon.$$
Then, letting $M$ denote Algorithm~\ref{alg:asvt} and $A$ denote the output of $M$,
\begin{align}
\nonumber \frac{\Pr[M(X)=A]}{\Pr[M(X^\prime)=A]}&=\prod_{i\in [|A|]}\frac{\Pr[M_i(X)=a_i|b_i]}{\Pr[M_i(X^\prime)=a_i|b_i]}\\ \nonumber
&\leq \varepsilon,
\end{align}
where $b_i$ is the ${\rm aborted}$ variable at $i$th round.
Therefore, Algorithm~\ref{alg:asvt} satisfies $(\varepsilon, p)$-ADP. 
\end{proof}
\fi

\section{Experiments}
\label{sec:exp}
The source code used in these experiments is available at 
\url{https://github.com/tkgsn/adp-algorithms}.

\subsection{Preliminaries}
First, we describe the detail of the dataset and queries.
Then, we explain the metrics to evaluate utility of our mechanisms.

\subsubsection{Dataset}
We use the real-world trajectory dataset, called Peopleflow\footnote{\url{http://pflow.csis.u-tokyo.ac.jp/}}.
For simulation, we assume that data owners are infected by some disease such as COVID-19.
A record (i.e., a trajectory) is a sequence of tuples (userID, latitude, longitude, \textit{placeID}, \textit{state}, \textit{timestamp}): \textit{placeID} represents the category of the location if the location has a category, \textit{state} indicates the movement (i.e., "STAY" or "MOVE"), and \textit{timestamp} is the time user visited the location.
There are $5{,}835$ locations on Tokyo, Japan in the Peopleflow dataset, and we assume that if a tuple includes a \textit{placeID} and has the "STAY" attribute, the data owner has visited the location.
We specify $(3,4,5,6,7)$ as placeID, which represents shops, restaurants, entertainments such as movie theaters and art galleries, groceries, schools, or other buildings.

\subsubsection{Query}
As described in Section~\ref{subsec:problem_setting}, we query the threshold propositions $q^i_\leq=p_\leq^t\circ f_{C^T}(\cdot)_i$ and $q^i_\geq=p_\geq^t\circ f_{C^T}(\cdot)_i$.
We use two types of counting query $f_{C^T}$: high-sensitive one and low-sensitive one according to $T$ (i.e., time or date).
Here, we explain the detail of them.
\paragraph{Low-sensitive $f_{C^T}$}
This query asks whether many people are not in close contact with each other in target locations.
That is, given target locations $(l_1,l_2,\dots,l_d)$ and time $T$, we let $C^T=(c^T_1,c^T_2,\dots,c^T_d)$, where $c^T_i$ is the condition that asks whether the user visited the location $i$ at \textbf{time} $T$ (i.e., time range is very small).
In this case, $f_{C^T}$ is low-sensitive because each user is able to be at only one location at the time $T$. 
That is, for two datasets $X,X^\prime$ such that $X\sim X^\prime$, $|f_{C^T}(X)-f_{C^T}(X^\prime)|\leq 1$.
\paragraph{High-sensitive $f_{C^T}$}
This query asks whether the target locations are not visited by many infected people.
That is, given target locations $(l_1,l_2,\dots,l_d)$ and given date $T$, we let $C^T=(c^T_1,c^T_2,\dots,c^T_d)$, where $c^T_i$ is condition that asks whether the user visited the location $l_i$ at \textbf{date} $T$ (i.e., time range is large).
In this case, $f_{C^T}$ is high-sensitive because each user can visit multiple locations at the date $T$.
That is, for two datasets $X,X^\prime$ such that $X\sim X^\prime$, $|f_{C^T}(X)-f_{C^T}(X^\prime)|\leq |C^T|$.

\subsubsection{Policy}
Given $C^T$, we use the following policy:
$
p(x):=\neg c^T_1(x) \lor \neg c^T_1(x)\dots \lor \neg c^T_d(x).
$
Since $\Delta_p(f_{C^T})$ is monotonically decreasing, ALap and SALap achieve one-sided error from Corollary~\ref{col:utility one-sided alap} for $q_\leq^i$ for all $i\in [d]$.
We assume that the information that a user did not visit target locations is not non-sensitive.
Under this assumption, this policy is reasonable from the privacy guarantee of $(\varepsilon,p)$-ADP as described in Section~\ref{subsec:privacy guarantee}.

\subsubsection{Evaluation}
Given $(p^t_\leq\circ f_{C^T}(\cdot)_i, \alpha,\beta)$-one-sided mechanism $M$,
We evaluate two types of utility: one-sided error ($1-\beta$) and the number of locations a mechanism can answer.
First, we define expected $\beta$ as follows: given $X\in\mathcal{X}^n$,
$$
\mathbb{E}[\beta|\leq]:= \mathbb{E}_{z\sim M(X), i\in\{i|p^t_\leq\circ f_{C^T}(\cdot)_i(X)\}}[\Pr[p_i^{\rm san}(z)={\rm True}]],
$$
where $p^{\rm san}_i$ is a $(p^t_\leq\circ f_{C^T}(\cdot)_i, \alpha,\beta)$-sanitizer (refer to Section \ref{sec:one-sided accuracy}).
For example, if $M={\rm ALap}_{p,\varepsilon,f_{C^T}}$, $p_i^{\rm san}(z):=p^t_\leq(z_i)$.
That is, $\mathbb{E}[\beta|\leq]$ is the expected probability that a mechanism rightly answers $q^i_\leq$ under the condition that $i$ is randomly chosen from $[d]$ such that $p^t_\leq\circ f_{C^T}(\cdot)_i(X)={\rm True}$ (i.e., condition that a target location is randomly chosen from safe locations).

\paragraph*{Competitors}
We compare our mechanism with a virtual mechanism that is $(\varepsilon,\delta)$-DP and optimal with respect to Inequality~(\ref{ineq:approximate dp condition}).
Blowfish privacy and OSDP are not subjects because undesirable privacy leaks occur as described in Sections~\ref{sec:insufficiency_blowfish} and \ref{sec:insufficiency_osdp}.

\subsection{$q_\leq$: Low-sensitive Case}
\label{sec:low-sensitive case}
Here, we consider the low sensitive query by setting $T$ on an hourly basis.

\subsubsection{Visualization}
Here, we set $T={\rm Dec}/22{\rm th}/2013$ $6$ p.m. and $t=5$, and we visualize the result on Figure~\ref{fig:result_example} in Introduction.
Location $l_i$ is marked as safe or dangerous if it is marked with True by a $(p^t_\leq\circ f_{C^T}(\cdot)_i,\alpha,\beta)$-one-sided mechanism or a $(p^t_\geq\circ f_{C^T}(\cdot)_i,\alpha,\beta)$-one-sided mechanism, respectively.
Obscure spots represent locations where the number of visited people is obscure because of noise (no guarantee of one-side error).
We can see that the difference is remarkable when $\alpha=0$.
$1$-DP does not allow one-sided error, so all outputs are obscure.
$(1,10^{-4})$-DP allows one-sided error, but almost all locations are obscure.
On the other hand, our mechanism can find $5{,}031$ safe locations out of $5{,}343$ locations.
This is due to the relaxation by $p$.

Next, we set $\alpha=10^{-3}$ and visualize the results on Figure~\ref{fig:hotspots_alpha_10^-4} to show that $(\varepsilon,\delta)$-DP still miss to find many safe and dangerous locations.
We can see that $1$-DP and $(1,10^{-4})$-DP generate almost the same results.
That is, relaxation by $\delta$ does not effectively work to achieve one-sided error.
Note that the result of our mechanism for dangerous locations are almost the same as other results because of privacy protection as shown in Section~\ref{sec:one-sided error for q_geq}.

\begin{figure*}[t]
 \begin{minipage}{0.24\hsize}
  \centering\includegraphics[width=\hsize]{new_imgs2/ground_truth.png}
 \end{minipage}
 \begin{minipage}{0.24\hsize}
  \centering\includegraphics[width=\hsize]{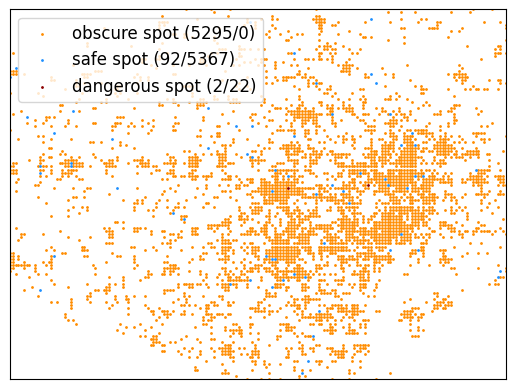}
 \end{minipage}
  \begin{minipage}{0.24\hsize}
  \centering\includegraphics[width=\hsize]{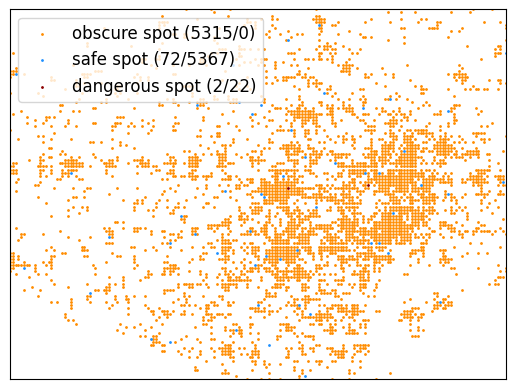}
 \end{minipage}
  \begin{minipage}{0.24\hsize}
  \centering\includegraphics[width=\hsize]{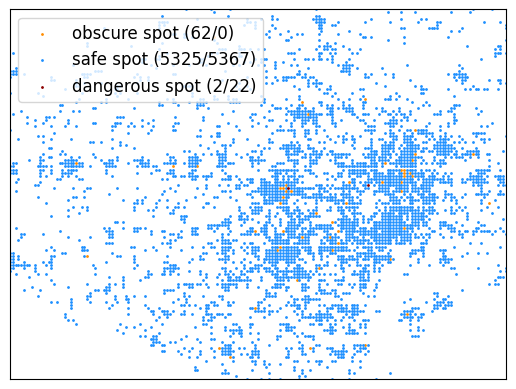}
 \end{minipage}
 \caption{Visualization when $T={\rm Dec}/22{\rm th}/2013$ $6$ p.m., $\varepsilon=1, \alpha=10^{-4}, \delta=10^{-4}$, and $t=5$ of ground truth, $\varepsilon$-DP, $(\varepsilon,\delta)$-DP, $(\varepsilon,p)$-ADP from left to right.  \label{fig:hotspots_alpha_10^-4}}
\end{figure*}

\subsubsection{Varying parameters}
Here, we set $\varepsilon=1, \alpha=10^{-3}, \delta=10^{-4}$, and $t=5$ as default.
We vary each parameter and show $\mathbb{E}[\beta|\leq]$ comparing with $(\varepsilon,\delta)$-DP.
We randomly select $T$ and we take the average of $\mathbb{E}[\beta|\leq]$.
We plot the results on Figure~\ref{fig:varying_parameters}.

\begin{figure*}[t]
 \begin{minipage}{0.24\hsize}
  \centering\includegraphics[width=\hsize]{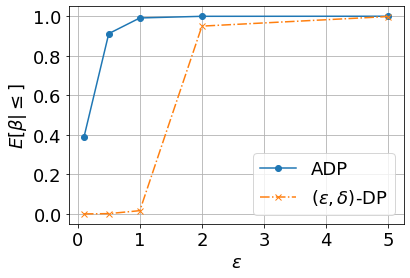}
 \end{minipage}
 \begin{minipage}{0.24\hsize}
  \centering\includegraphics[width=\hsize]{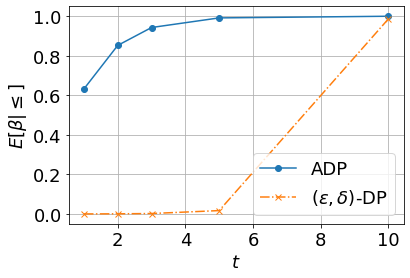}
 \end{minipage}
  \begin{minipage}{0.24\hsize}
  \centering\includegraphics[width=\hsize]{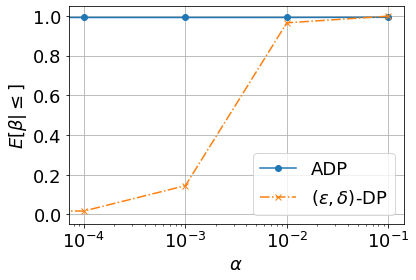}
 \end{minipage}
  \begin{minipage}{0.24\hsize}
  \centering\includegraphics[width=\hsize]{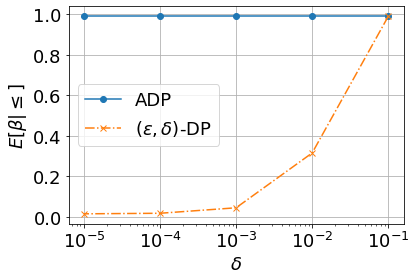}
 \end{minipage}
 \caption{$\mathbb{E}[\beta|\leq]$ when varying $\varepsilon$, $t$, $\alpha$, and $\delta$. $T$ is randomly chosen and we take the average.\label{fig:varying_parameters}}
\end{figure*}

While $(\varepsilon,\delta)$-DP sharply decreases when $\varepsilon$ and $t$ becomes smaller, we can see that ADP is robust to small $\varepsilon$ and small $t$.
For example, in the case of $t=1$ (i.e., judging whether one was not in close contact with an infected person), $(\varepsilon,\delta)$-DP does not work at all; on the other hand, our mechanism can accurately judge a not close contact with probability about $0.4$.

Next, we check how $\alpha$ or $\delta$ is required for $(\varepsilon,\delta)$-DP to achieve one-sided error.
Note that ADP is not affected by $\alpha$ and $\delta$.
When $\alpha=10^{-2}$, $(\varepsilon,\delta)$-DP becomes close to ADP, but $\alpha=10^{-2}$ is too large for our setting.
This is because if we get $5{,}000$ answers like Figure~\ref{fig:hotspots_alpha_10^-4} for example, it includes $50$ errors in average.
Also, when $\delta=10^{-1}$, $(\varepsilon,\delta)$-DP becomes close to ADP, but $\delta=10^{-1}$ is clearly unacceptable value.

\paragraph*{Remark}
It should be noted that the number of target locations (i.e., $|C^T|$) does not affect one-sided error in the low-sensitive case.
However, the qualitative privacy of the policy becomes weaker from the implication of Proposition~\ref{prop:leaked_information}.

\subsection{$q_\leq$: High-sensitive Case}
Here, we consider high sensitive query by setting $T$ on an date basis.
In this case, we use SALap with $p_i^{\rm san}(A):=(A_i=\perp)$, which is robust to the sensitivity.
Here, we evaluate the number of right answers (i.e., $|\{a\in A|a=\perp\}|$ where $A={\rm SALap(X)}$) in addition to $\mathbb{E}[\beta|\leq]$.

\subsubsection{Visualization}
First, we visualize how many locations SALap can rightly answer in Figure~\ref{img:nquery_cumulative} as cumulative probabilities.
We set $T={\rm Dec}/22{\rm th}/2013$, $\alpha=0$, $t=5$, and $\varepsilon=1$.
While ALap answers only one location, for example, SALap can get at least $10$ right answers (i.e., $l_1,l_2,\dots,l_{10}$) with about probability of about $0.5$.
This difference comes from the composed technique of SALap.

\subsubsection{Varying parameters}
We set $\alpha=10^{-3}, \epsilon=1, \delta=10^{-4}$, and $t=5$ as default.
Note that we omit the case varying $\delta$ because $(\varepsilon,\delta)$-DP does not work well even low-dimensional case.
We can see that $\mathrm{E}[\beta|\leq]$ is almost the same as the low-dimensional case.
However, unlike the low-dimensional case, the number of answered queries greatly affects $\varepsilon$ and $t$; when $\varepsilon<1$ or $t< 3$, the number of right answers sharply decreases.
In the high-dimensional case, we need to spend great privacy budget or run SALap with a large threshold setting to target many locations.

\begin{figure*}[t]
 \begin{minipage}{0.24\hsize}
  \centering\includegraphics[width=0.9\hsize]{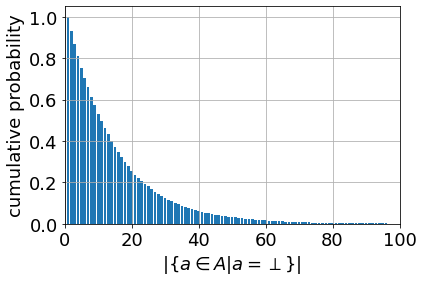}
  \caption{Visualization of $|\{a\in A|a=\perp\}|$ when $\alpha=0$, $t=5$, $\varepsilon=1$, and $T={\rm Dec}/22{\rm th}/2013$. \label{img:nquery_cumulative}}
 \end{minipage}
 \begin{minipage}{0.75\hsize}
 \begin{minipage}{0.32\hsize}
  \centering\includegraphics[width=\hsize]{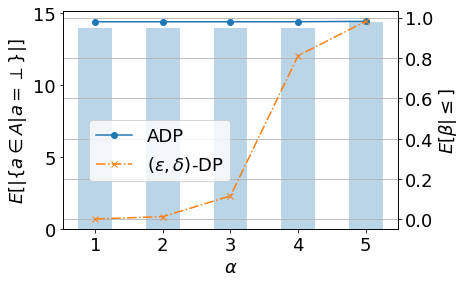}
 \end{minipage}
  \begin{minipage}{0.32\hsize}
  \centering\includegraphics[width=\hsize]{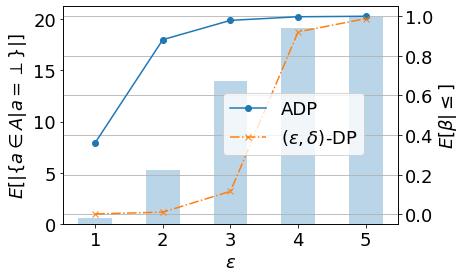}
 \end{minipage}
  \begin{minipage}{0.32\hsize}
  \centering\includegraphics[width=\hsize]{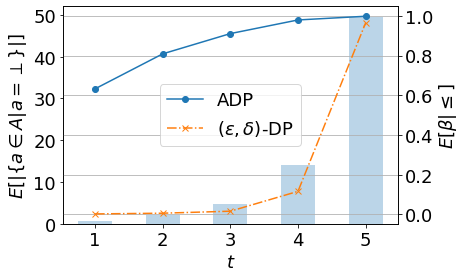}
 \end{minipage}
 \caption{$\mathbb{E}[\beta|\leq]$ (line graph) and $\mathbb{E}_A[\{a\in A|a=\perp\}]$ (bar graph) when varying parameters of $\alpha$, $\varepsilon$, and $t$. $T$ is randomly chosen and we take the advantage of results.}
 \end{minipage}
\end{figure*}

\section{Conclusion}
In this paper, we proposed a new privacy definition called ADP by relaxing DP with the belief to achieve one-sided error.
We showed that the privacy guarantee of ADP in our case is reasonable based on Bayesian analysis.
We also proposed two ADP mechanisms called ALap and SALap for low-sensitive and high-sensitive cases.
Finally, we conducted the experiments and show the practicality for epidemic surveillance using real-world dataset.

We found some limitations of ADP.
First, ADP results in unreasonable privacy protection to answer $q^i_\geq$ with one-sided error.
Second, when the sensitivity is high and $t$ and $\varepsilon$ are small, it is difficult to answer many queries with practical one-sided error.
Especially for dense data (many non-zero values) decreases number of answers of SALap.

This paper focuses on location monitoring, but one-sided error is crucial for other subjects (e.g., outlier detection and medical care).
Also, ADP is the general notion of privacy, so we believe that ADP effectively works more general purposes other than one-sided error.
These are interesting future directions.



\bibliographystyle{IEEEtran}
\bibliography{IEEEexample}

\iflong
\section{Appendices}
\label{sec:geometric_mechanism}
\subsection{The Geometric Mechanism}
\label{subsec:geometric_mechanism}
The geometric mechanism~\cite{ghosh2012universally} is a discrete mechanism for counting query $f:\mathcal{X}^n\to\{0,1,\dots,n\}^d$.
The mechanism adds noise like the Laplace mechanism, but $\lambda_i$ is i.i.d., and following $\frac{1-\mathrm{e}^{-\varepsilon}}{1+\mathrm{e}^{-\varepsilon}}\exp{(\frac{-\varepsilon |\lambda|}{\Delta_f})}$ for all $i\in[d]$.
The geometric mechanism satisfies $\varepsilon$-DP.

\subsection{Asymmetric geometric mechanism}
\begin{definition}[Asymmetric Geometric mechanism]
Given policy $p$, privacy parameter $\varepsilon$, and counting query $f:\mathcal{X}^n\to\{0,1,\dots,n\}^d$ whose $p$-sensitivity is $1$, the asymmetric geoemtric mechanism $AGeo_{p,\varepsilon, f}$ is:
$$AGeo_{p,\varepsilon, f}(X)=(\lambda_1,\dots,\lambda_d)$$
where $\lambda_i$ is independently distributed, and is following the discrete distribution for all $i\in[d]$:\\
If $p$-sensitivity of $f(\cdot)_i$ is monotonically decreasing,
\begin{align}
\nonumber \left\{ \begin{array}{ll}
\mathrm{e}^{-(\lambda-f(X))\varepsilon} & (\lambda=n) \\
(1-\mathrm{e}^{-\varepsilon})\mathrm{e}^{-(\lambda-f(X))\varepsilon}& (\lambda\geq f(X)) \\
0 & (otherwise) \\
\end{array} \right.
\end{align}
else if $P$-sensitivity of $f(\cdot)_i$ is monotonically increasing,
\begin{align}
\left\{ \begin{array}{ll}
\nonumber \mathrm{e}^{(\lambda-f(X))\varepsilon} & (\lambda=0) \\
(1-\mathrm{e}^{-\varepsilon})\mathrm{e}^{(\lambda-f(X))\varepsilon}& (\lambda\leq f(X)) \\
0 & (otherwise) \\
\end{array} \right.
\end{align}
otherwise
$$\frac{1-\mathrm{e}^{-\varepsilon}}{1+\mathrm{e}^{-\varepsilon}}\mathrm{e}^{-\varepsilon |\lambda-f(X)|}$$
\end{definition}

\begin{theorem} 
\label{theo:osdp geo privacy}
$\forall p, \varepsilon\in\mathbb{R}^+, f:\mathcal{X}^n\to\{0,1,\dots,n\}^d$, AGeo$_{p,\varepsilon,f}$ satisfies ($\varepsilon,p$)-ADP\@.
\end{theorem}
The proof of this is analogous of that of Theorem~\ref{theo:adp_laplace_privacy}.

\paragraph{Optimality}
Formally, the optimality of AGeo is as follows:
\begin{theorem}
\label{theo:optimality generalized geo}
Given policy $p$ and one-dimensional counting query $f:\mathcal{X}^n\to\{0,1,\dots,n\}$ whose $p$-sensitivity is $1$ and monotonic, the AGeo has the lowest error.
Here, error is:
$$\mathbb{E}_{X\sim\pi}[\mathbb{E}_{z_X}[loss(z_X, f(X))]]$$
where $\pi$ is a prior distribution of $X$ and $z_X$ is a random variable output from $m(X)$, and loss$(z_X, f(X))$ is an arbitrary measure subject to being non-negative and non-decreasing in $|m(X)-f(X)|$ (e.g., mean absolute error).
\end{theorem}

\label{sec:optimality_of_geometric_mechanism}
\begin{proof}
Here, we prove the theorem in the case where $p$-sensitivity is monotonically decreasing.
When $p$-sensitivity is monotonically increasing, the proof is analogous.

We consider the distribution of the mechanism as $(n+1)\times (n+1)$ matrix $M$, which we call a distribution matrix.
Here, $M_{i,j}$ represents the probability outputting $j$ when the answer of the counting query is $i$.
Then, the following lemmas hold.
\begin{lemma}
\label{lemma:c_upper_triangle}
The distribution matrix of the mechanism optimal for mean absolute error is an upper triangle matrix.
\end{lemma}

\begin{proof}[Proof of Lemma~\ref{lemma:c_upper_triangle}]
We assume that there is an optimal mechanism such that the distribution matrix $M$ is not an upper triangle matrix, and derive a contradiction.

We assume $M_{i,j}\neq 0$ where $i>j$ and $i>0$.
Since the mechanism satisfies $(\varepsilon,p)$-ADP, $M_{i-1,j}\geq \mathrm{e}^{-\varepsilon}M_{i, j}>0$.
Iteratively applying this statement, we get $M_{j+1, j}>0$.
By moving $M_{j+1, j}$ to $M_{j+1, j+1}$ (i.e., setting $M_{j+1, j+1}\Leftarrow M_{j+1,j+1}+M_{j+1, j}$ and $M_{j+1, j}\Leftarrow 0$), the updated mechanism has a lower error than the original mechanism.
However, this mechanism does not satisfy $(\varepsilon,p)$-ADP, so we move $M_{j+k, j}$ to $M_{j+k, j+1}$ for all $k\in[2,i-j]$.
This updated mechanism satisfies $(\varepsilon,p)$-ADP and has a lower error than the original mechanism.
This contradicts the fact that the original mechanism is optimal.
\end{proof}

\begin{lemma}
\label{lemma:constraint_in_upper_triangle}
The distribution matrix $M$ of the mechanism optimal for absolute mean error satisfies $\forall i,j\in[n]$ such that $i<j$,
\begin{equation}
\label{eq:relationship_in_c}
M_{i,j}=\mathrm{e}^{-\varepsilon}M_{i+1, j}
\end{equation}
\end{lemma}

\begin{proof}[Proof of Lemma~\ref{lemma:constraint_in_upper_triangle}]
Also, we derive a contradiction by assuming the optimal mechanism whose distribution matrix does not satisfy Equality \ref{eq:relationship_in_c}.
From the constraint of ADP, 
$$M_{i,j}>\mathrm{e}^{-\varepsilon}M_{i+1, j}$$
Putting $D = M_{i,j}-\mathrm{e}^{-\varepsilon}M_{i+1, j}$, we set $M_{i,i}\Leftarrow M_{i,i}+D$.
Now, $D>0$ from the assumption, so the updated mechanism has a lower error.
The constraint of $M_{i,i}$ is $M_{i,i}\geq 0$ since $M_{i+1, i}=0$ from Lemma~\ref{lemma:c_upper_triangle}.
Therefore, the updated mechanism satisfies $(\varepsilon,P)$-ADP.
This is the contradiction to the fact that the original mechanism is optimal.
\end{proof}

Now, we can prove Theorem~\ref{theo:optimality generalized geo} using Lemma~\ref{lemma:c_upper_triangle} and \ref{lemma:constraint_in_upper_triangle}.
From Lemma~\ref{lemma:c_upper_triangle}, $M_{n+1,n+1}=1$ since $M_{n+1,\cdot}$ is a probability distribution.
Then, by inductively applying Lemma~\ref{lemma:constraint_in_upper_triangle} from $M_{n+1,n+1}=1$, we get the distribution matrix of the ADP mechanism.
\end{proof}

\subsection{The composition theorem}
\label{sec:composition}
Here, we derive the composition theorem for ADP.

\begin{theorem}
\label{theo:generalized OSDP composition}
We consider the sequential mechanism $M_i:\mathcal{Z}^{(1)}\times\dots\times\mathcal{Z}^{(i-1)}\times\mathcal{X}^n\to\mathcal{Z}_i$ for $i\in[k]$ where $\mathcal{Z}_i$ is the range space of $M_i$.
We assume that for all $i\in[k]$, $M_i(z_{1:i-1},\cdot)$ satisfies $(\varepsilon_i,p_i)$-ADP for any value of auxiliary input $z_{1:i-1}$.
Then, mechanism $M:\mathcal{X}^n\to\mathcal{Z}_1\times\dots\times\mathcal{Z}_k$ sequentially applying $M_i$ satisfies $(\sum_{i=1}^k\varepsilon_i, \sum p_i)$-ADP\@.
\end{theorem}
\begin{proof}
From the definition of ADP, it holds that $\forall i\in[k], S_i\subseteq \mathcal{Z}^{(i)},X\sim^{\sum_i p_i}X^\prime,$
\begin{equation}
\nonumber
    \Pr[M_i(z_{1:i-1},X)\in S_i]\leq\mathrm{e}^{\varepsilon_i}\Pr[M_i(z_{1:i-1},X^\prime)\in S_i]
\end{equation}
Therefore, $\forall S\subseteq \mathcal{Z}_1\times\dots\times\mathcal{Z}_k$, the following inequality holds.
\begin{align}
\nonumber
    \frac{\Pr[M(X)\in S]}{\Pr[M(X^\prime)\in S]}&=\frac{\Pi_{i=1}^k\Pr[M_i(z_{1:i-1},X)\in S_i]}{\Pi_{i=1}^k\Pr[M_i(z_{1:i-1},X^\prime)\in S_i]}\\
    \nonumber
    &\leq\mathrm{e}^{\sum_{i=1}^k{\varepsilon_i}}.
\end{align}

\end{proof}

This composition theorem says that not only the privacy parameter sequentially adds up as DP, but also the composition of two mechanisms with different policies results in a mechanism with a policy constructed by adding the two policies.
\fi

\end{document}